\documentclass[conference]{IEEEtran} % final mode\IEEEoverridecommandlockouts
\usepackage[utf8]{inputenc} % allow utf-8 input
\usepackage[T1]{fontenc}    % use 8-bit T1 fonts
\usepackage{cite}
\usepackage[hidelinks]{hyperref}       % hyperlinks
\usepackage[table]{xcolor}% http://ctan.org/pkg/xcolor
\usepackage{dsfont}

\usepackage{booktabs}
\usepackage{xcolor,colortbl}
\definecolor{lavender}{gray}{0.95}

\usepackage{amsmath,amssymb,amsfonts}
\usepackage{graphicx}
\usepackage{textcomp}
\usepackage{xcolor}
\usepackage{comment}
\usepackage{amsmath,amsthm,amsfonts,amssymb,amscd}
\usepackage{siunitx}

\newtheorem{theorem}{Theorem}

\usepackage{blindtext}
\usepackage{abraces}
\usepackage{float}
\usepackage{resizegather}
\usepackage{nicefrac}
\usepackage{graphicx,subcaption}
\usepackage{multirow}
\definecolor{tgray}{gray}{0.95}
\usepackage{algorithm}
\usepackage{algpseudocode}

\definecolor{color4}{RGB}{179, 43, 59}
\definecolor{green2}{rgb}{0, 153, 153}
\definecolor{wpurple}{rgb}{.341, .035, .953}
\definecolor{wblue}{RGB}{0,175,187}

\definecolor{wred}{RGB}{252,78,7}

\usepackage{listings}

\newtheorem*{remark}{Remark}

\newif\ifbulletlist
\newif\iftext
\texttrue % comment to hide text

\def\BibTeX{{\rm B\kern-.05em{\sc i\kern-.025em b}\kern-.08em
    T\kern-.1667em\lower.7ex\hbox{E}\kern-.125emX}}
    
\begin{document}

%\title{Optimal Dispatching Strategies for Platooning Formation at Controllable Stations*\\

\title{Optimal and Heuristic Approaches for Platooning Systems with  Deadlines
}%\author{Anonymous\thanks{} \and \quad\quad    Anonymous$^*$ \and \quad\quad Anonymous\thanks{} \and \quad\quad Anonymous\thanks{} \and \quad\quad Anonymous\thanks{}  \\~\\~\\}
\author{
\IEEEauthorblockN{
Thiago S. Gomides\IEEEauthorrefmark{1},
Evangelos Kranakis\IEEEauthorrefmark{1},
Ioannis Lambadaris\IEEEauthorrefmark{2},
Yannis Viniotis\IEEEauthorrefmark{3},
Gennady Shaikhet\IEEEauthorrefmark{4}
}

\IEEEauthorblockA{\IEEEauthorrefmark{1}School of Computer Science, Carleton University, Ottawa, ON, Canada}

\IEEEauthorblockA{\IEEEauthorrefmark{2}Department of Systems and Computer Engineering, Carleton University, Ottawa, ON, Canada}

\IEEEauthorblockA{\IEEEauthorrefmark{3}Department of Electrical and Computer Engineering, North Carolina State University, Raleigh, NC, USA}

\IEEEauthorblockA{\IEEEauthorrefmark{4}School of Mathematics and Statistics, Carleton University, Ottawa, ON, Canada}

Email: \{thiagodasilvagomides@cmail, kranakis@scs, ioannis@sce, gennady@math\}.carleton.ca, candice@ncsu.edu.}

\maketitle
\ifbulletlist{\color{blue} 
\begin{enumerate}

{\color{red}\item  MISSING MOTIVATON: importance of the problem
\item Missing "marketing"; difficulty.
}

\item Problem Context

\item System Description

\item Objective

\item Model Formulation

\item Analytical Scope

\item Structural Properties

\item Computational Challenge

\item Proposed Solution
\end{enumerate}
}
\fi

\begin{abstract}
Efficient truck platooning is a key strategy for reducing freight costs, lowering fuel consumption, and mitigating emissions. Deadlines are critical in this context, as trucks must depart within specific time windows to meet delivery requirements and avoid penalties. In this paper, we investigate the optimal formation and dispatch of truck platoons at a highway station with finite capacity $L$ and deadline constraints $T$. The system operates in discrete time, with each arriving truck assigned a deadline of $T$ slot units. The objective is to leverage the efficiency gains from forming large platoons while accounting for waiting costs and deadline violations. We formulate the problem as a Markov decision process and analyze the structure of the optimal policy $\pi^\star$ for $L = 3$, extending insights to arbitrary $L$. We prove certain monotonicity properties of the optimal policy in the state space $\mathcal{S}$ and identify classes of unreachable states. Moreover, since the size of $\mathcal{S}$ grows exponentially with $L$ and $T$, we propose heuristics--including conditional and deep-learning based approaches--that exploit these structural insights while maintaining low computational complexity.
\end{abstract} \section{Introduction and Motivation}
\ifbulletlist{\color{blue} 
\begin{enumerate}
    \item Define truck platooning and highlight its potential benefits (three benefits).
    \item Emphasize the importance of analyzing platooning under rigorous mathematical models.
\end{enumerate}
}
\fi

Platooning refers to vehicle convoys that travel in close formation, similar to a train or motorcade. In their simplest form, platoons can form naturally on busy roads~\cite{d243b67ca2da4b7e88c224fa5f0ce3af}. In practice, however, maintaining such convoys requires advanced communication and automation technologies~\cite{Hao04052025,BALADOR2022100460,TSUGAWA201341}.

In this work, we are particularly interested in the formation of \textit{truck platoons}---an effective strategy for reducing freight costs, especially fuel consumption~\cite{TSUGAWA201341, ZHANG2024105106, 11161316}, while lowering greenhouse gas emissions and improving highway safety~\cite{Alvarez01071999}.

Effective \textit{coordination} is essential to realizing these benefits: trucks must be grouped and dispatched to maximize platoon size~\cite{ZHANG2017357} while respecting deadlines~\cite{9944383, 9102259}, vehicle-level requirements, and technological constraints~\cite{LARSSON2015258}.

Motivated by these challenges, we study the optimal formation of truck platoons at highway stations (e.g., gas stations or rest areas) with finite capacity $L$ and deadline constraints~$T$.

\subsection{Related Work}
\ifbulletlist{\color{blue} 
\begin{enumerate}
    \item Review of recent research on platooning with deadlines. (2-3 papers)
    \item Discussion of relevant related work, including \url{https://ieeexplore.ieee.org/stamp/stamp.jsp?tp=\&arnumber=9944383} (3rd or 4th).
\end{enumerate}
}\fi %bulletlist

Platooning with deadlines has been primarily studied in two domains: 1) \textit{path planning}~\cite{ZHANG2017357, 9944383, 9102259} and 2) \textit{vehicle routing}~\cite{LARSSON2015258}. In path planning, the focus is on continuous trajectory control---adjusting speed and spacing so that trucks can merge into platoons while meeting time requirements. In vehicle routing, the objective is to determine optimal routes and departure schedules across a network. Both domains aim to form energy-efficient platoons while respecting delivery deadlines.

In~\cite{ZHANG2017357}, the authors study the formation of a two-truck platoon with stochastic arrivals at a highway station. A platoon forms if both trucks arrive simultaneously; otherwise, one truck must wait, incurring delay costs. The optimal policy forms a platoon only when the waiting time does not exceed a specified deadline. In~\cite{9944383}, a similar problem is considered, extended to include \textit{speed planning}. In~\cite{9102259}, deep reinforcement learning at an edge node is used to optimize platooning opportunities, with the goal of maintaining platoon stability, minimizing fuel consumption, and satisfying deadlines.

In~\cite{LARSSON2015258}, platoon routing is examined for trucks with deadlines traveling across stations in a road network. The problem is modelled as a graph-routing problem and solved using integer linear programming. Optimal solutions are computed for small instances, while three heuristics handle larger scenarios. %Coordinated routing can achieve roughly 9--10\% fuel savings while satisfying delivery deadlines.

\subsection{Novelty and Main Contributions}
\ifbulletlist{\color{blue} 
\begin{enumerate}
    \item A novel cost function and system model.
    \item First platooning framework with deadline extendable to arbitrary $L$ and $T$.
    \item Heuristics derived from structural properties of the optimal policy.
    \item Benchmarking of policies and analysis of computational overhead.
\end{enumerate}
}\fi %bulletlist

In this work, we study the optimal formation of truck platoons under deadline constraints. Similar to~\cite{ZHANG2017357, 9944383, LARSSON2015258} (and unlike~\cite{9102259}), we consider platoon formation at a highway station with stochastic truck arrivals. Unlike~\cite{ZHANG2017357, 9944383, 9102259}, we do not restrict platoon size. As in~\cite{ZHANG2017357}, we formulate an optimal control problem and solve it using dynamic programming (DP).  Similar to~\cite{LARSSON2015258}, we design scalable heuristics to overcome the computational challenges of DP.

The main contributions of this work are as follows:
\begin{itemize}
%\item We formalize the optimal platooning coordination problem for arbitrary station capacity $L$ and deadline $T$.
\item For station size $L = 3$, we prove \textit{monotonicity properties} of the optimal policy $\pi^\star$ and identify \textit{unreachable states}.
\item We generalize insights from the $L = 3$ case to arbitrary~$L$.
\item We propose and numerically evaluate scalable, near-optimal heuristics that exploit structural insights of $\pi^\star$.
\end{itemize}
The paper is organized as follows. In Section~\ref{sec:formulation},  we formulate the platooning model, and in Section~\ref{sec:control},  we introduce the control problem. In Section~\ref{sec:characterization}, we characterize the expected average cost and the optimal policy, while in Section~\ref{sec:heuristics}, we present our heuristic policies. Numerical results are presented in Section~\ref{sec:performance}. We conclude  with a discussion of future work in Section~\ref{sec:conclusion}. % Due to space constraints, formal proofs are provided in~\cite{arxiv}; here, we focus on the key intuitions underlying each theorem.

\section{Platooning Model}\label{sec:formulation}
\ifbulletlist{\color{blue} 
\begin{enumerate}
    \item Description of the station where truck platoons are formed. Cost structure associated with station operations.
    \item Controller in charge of forming platoons. Motivation for cost control.
    \item Station operations: deadline assignment and countdown.
    \item Actions available to the controller.
\end{enumerate}
}\fi

\begin{comment}
Consider a highway station with finite capacity $L$, where arriving trucks may form platoons prior to departure, as illustrated in Figure~\ref{fig:overview}.
\begin{figure}[H]
\centering
\includegraphics[width=\linewidth]{figures/overview.pdf}
\caption{System model illustration.}
\label{fig:overview}
\end{figure}

Each truck arriving at time $t$ is assigned a deadline $t+T$, which specifies the latest time by which it must depart. If a truck is not included in a platoon by this deadline, it is forced to leave individually. The primary objective of the system is to capitalize on the efficiency gains of forming full platoons (ideally of size $L$), while adhering to deadline constraints. As deadlines approach, however, the controller may need to dispatch partial platoons (of size $\ell < L$) in order to avoid the higher costs associated with deadline expirations. Moreover, trucks accumulate dwell costs while waiting at the station.

At each decision epoch, the controller chooses between two actions: \emph{release} or \emph{hold}. The control problem is therefore to balance the trade-off between operational costs (waiting costs and expiration penalties) and the efficiency benefits of dispatching full platoons.
\end{comment}

Consider a highway station with finite capacity to hold $L$ trucks, where arriving trucks may form platoons prior to departure. This system favours the formation of platoons of size $L$ (i.e., \emph{full platoons}) because they are more cost-efficient. Trucks incur dwell costs while waiting at the station, and penalties apply if they exceed their deadlines. Figure~\ref{fig:overview} illustrates this system.
\begin{figure}[htpb]
\centering
\includegraphics[width=\linewidth]{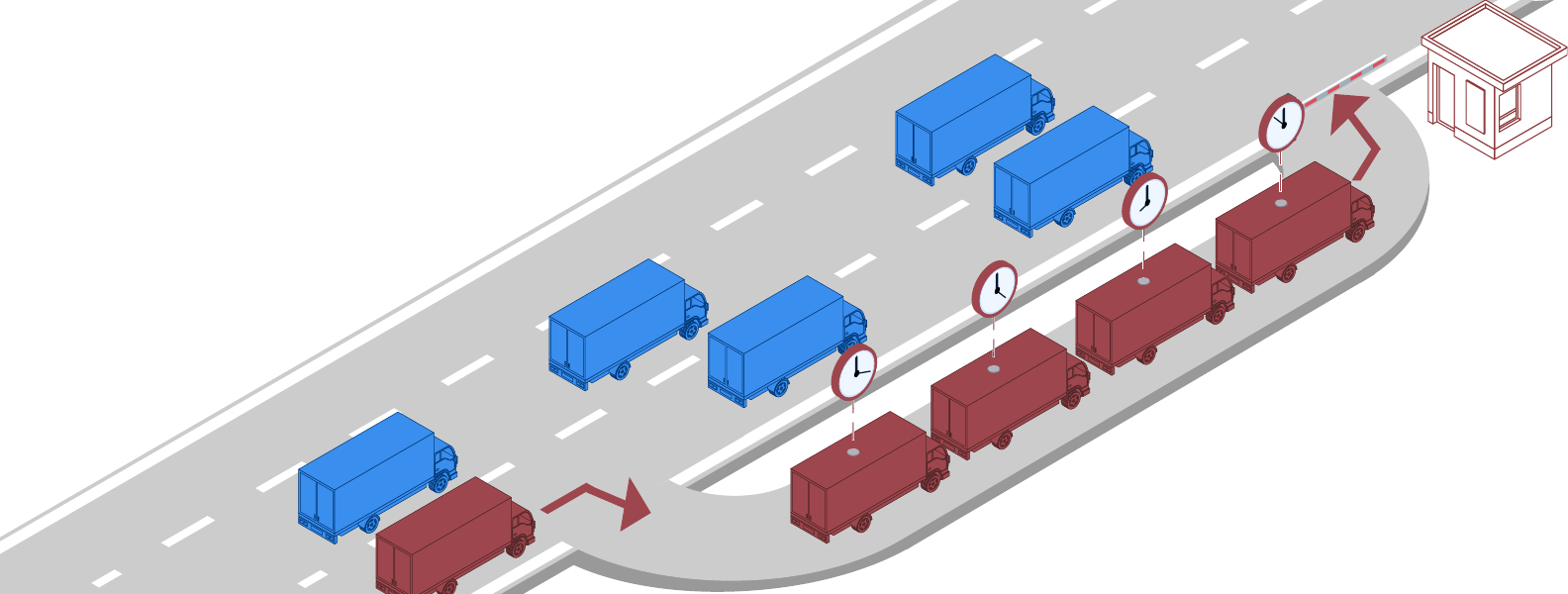}
\caption{System model illustration.}
\label{fig:overview}
\end{figure}

A centralized controller manages the station's operations, aiming to balance the efficiency gains from full platoons against the costs of waiting and potential deadline expirations.

Each arriving truck is assigned $T$ credits (in time units) upon arrival, representing its \textit{deadline}. Credits are decremented over time, and trucks with zero credits must depart from the station.

At each decision epoch, the controller chooses between two actions: to \emph{release} the waiting trucks as a platoon or to \emph{hold} them and wait for additional arrivals.

\section{Control Problem Formulation}\label{sec:control}
\ifbulletlist{\color{blue} 
\begin{enumerate}
	\item Need an opening statement for the general problem (point it to section VI-VII) and $L=3$ as a base case
    \item Description of the time model (discrete slots).
    \item Assumptions on the arrival process (Bernoulli with parameter $p$).
    %\item Deadline mechanism: trucks receive $T$ credits upon arrival, denoting maximum waiting time before forced departure.
\end{enumerate}
}\fi

For simplicity, in this section we focus on the case where the station capacity is fixed at $L = 3$. The general case with arbitrary $L$ is discussed in Sections~\ref{sec:heuristics}--\ref{sec:performance}. The corresponding platooning system for $L=3$ can be illustrated as in Figure~\ref{fig:deadlines1}.\begin{figure}[htpb]
  \centering
\includegraphics[width=.6\linewidth]{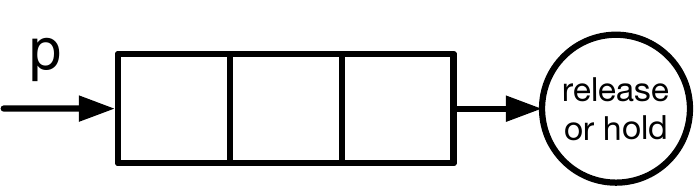}
  \caption{System model illustration for $L=3$.}
  \label{fig:deadlines1}
\end{figure}

%Time is modelled in $N$ discrete slots, where each slot $n \in N$ represents a fixed, uniform interval of real time (e.g., a few minutes). During each slot, new trucks may arrive according to a Bernoulli process with probability $0<p<1$.  Upon arrival, each truck is assigned $T \in \mathbb{Z}^+$ credits, representing its remaining deadline. At the start of each time slot, deadlines are decremented by one to account for time progression.

Time is modelled in discrete slots indexed by $n \in N$, each representing a fixed and uniform interval of real time (e.g., a few minutes). Deadlines are expressed in slot units and are decremented at the beginning of each slot.

A truck arrives with probability $p\in(0,1)$, independently across slots, according to a Bernoulli arrival process. 

 %and the controller can make dispatch decisions accordingly.
 %Upon arrival, each truck is assigned $T$ credits, representing the maximum number of time slots it can wait at the station before being forced to depart. This deadline mechanism is central to the control problem, as it defines the constraints under which the controller must decide which trucks or platoons to release at each time slot.

\subsection{Markov Decision Process (MDP) Formulation}\label{sec:markov}
\ifbulletlist{\color{blue} 
\begin{enumerate}
    \item State representation ($s$: remaining deadlines per truck; strictly increasing, related to Catalan numbers).
    \item Event definition (clean up phase and arrivals).
    \item Action set ($a_n = 0,1$).
    \item Mapping of events, actions, and cost computation within a time slot.
    \item Instantaneous cost function.
    \item State transition function $f(s,a,e)$.
\end{enumerate}
}\fi

\subsubsection{States} The system state is represented by the state vector $
	s = (d_3,\, d_2,\, d_1),$ where each $d_i \in \{1, \dots, T, \infty\}$ denotes the remaining deadline of the truck in position $i$, and $d_i = \infty$ indicates that the $i$-th position is unoccupied. 

By construction, the order of the deadlines is strictly decreasingly with respect to position, that is, 
\begin{align}
d_{i+1} > d_{i}, \quad \text{for all } d_{i} < \infty,
\label{eq:assumption}
\end{align}
 which ensures that trucks in lower-index positions always have earlier deadlines.  
 
Let $\mathcal{S}$ denote the set of all valid system states, defined by all feasible combinations of deadlines for waiting trucks. The cardinality of this set, $|\mathcal{S}|$, is given by
\begin{align}
|\mathcal{S}| = \sum_{k\,=\,0}^{L-1} \binom{T}{k} + \binom{T-1}{L-1},
\label{eq:catalan}
\end{align}
which is characterized by the \textit{Catalan numbers}~\cite{Stanley_2015}.

%{\color{red} transition of events "clean-up"=> decrement deadlines $(d_3,d_2,d_1) => (d_3-1,d_2-1,d_1-1)$. arrival if i'm state $d_3,d_2,d_1$ remind the infinity case}

\subsubsection{Events}

 %Each time slot consists of two consecutive events: \emph{deadlines decrement} and \emph{arrival}. 
%During deadlines decrement, all deadlines are decremented by one, and any truck whose deadline reaches zero departs the system. The remaining trucks shift right to preserve the strictly increasing order of deadlines in the state vector, with empty positions denoted by $\infty$. 

Each time slot consists of two consecutive events: (\textit{i}) a \emph{deadline decrement} and (\textit{ii}) an \emph{arrival}. During the deadline decrement, the state vector $(d_3,\, d_2,\, d_1)$ evolves to $(d_3-1,\, d_2-1,\, d_1-1)$. If $d_1-1 = 0$, the corresponding truck departs from the system. For components with $d_i = \infty$, we also have $d_i-1 = \infty$, in accordance with Eq.~\eqref{eq:assumption}. To maintain strictly decreasing order, the remaining components are then shifted to the right (i.e., $d_{i+1} \rightarrow d_i$), and $d_L = \infty$.

During the arrival phase, a new truck may enter the system. Let \(e_n = \mathds{1}\{\text{truck arrives at slot } n\}\), with $\mathbb{P}(e_n = 1) = p$ and $\mathbb{P}(e_n = 0) = 1 - p$. If $e_n = 1$, $s$ is updated by inserting a new component with deadline $T$ into the lowest available index.

For instance, if $s = (\infty,\, \infty,\, d_1)$ and $e = 1$, the updated state after both events occurs becomes $s' = (\infty,\, T,\, d_1-1)$.

%\begin{align}
%\mathbb{P}(e_n = 1) = p, \quad \mathbb{P}(e_n = 0) = 1 - p.
%\end{align}

\subsubsection{Actions} At each decision epoch, the controller selects an action $a_n$, where
\begin{align}
a_n =
	\begin{cases}
		0, & \text{\text{(hold)} trucks at the station}, \\
		1, & \text{\text{(release)} trucks as a platoon}.
	\end{cases}
	\label{eq:actions}
\end{align}

%\begin{remark}
%An expired trucks depart automatically during the deadline decrement phase. Since this occurs exogenously, it is not a controller decision and is excluded from the action set.\end{remark}

\subsubsection{Transition Probabilities}

Let $f(s, a, e)$ denote the mapping of the next state after applying action $a$ in state $s$, given the arrival indicator $e$. In the following expressions, $d_i = \infty$ for all $i$, unless stated otherwise. For $L = 3$, we have:{\begin{align*}
&f(s, a, e) =\\
&\begin{cases}
(\infty, \infty, \infty), & \text{if } a = 1, \forall s,   \forall e, \\
(\infty, \infty, \infty), & \text{if } a = 0,\, e = 0,\, d_1 = 1, \\
(\infty, \infty, T), & \text{if } a = 0,\, e = 1,\, d_1 = 1, \\
(\infty, \infty, d_2 - 1), & \text{if } a = 0,\, e = 0,\, d_1 = 1,\, d_2 \leq T, \\
(\infty, T, d_2 - 1), & \text{if } a = 0,\, e = 1,\, d_1 = 1,\, d_2 \leq T, \\
(\infty, \infty, d_1 - 1), & \text{if } a = 0,\, e = 0,\, d_1 \in (1,T], \\
(\infty, T, d_1 - 1), & \text{if } a = 0,\, e = 1,\,  d_1 \in (1,T], \\
(\infty, d_2 - 1, d_1 - 1), & \text{if } a = 0,\, e = 0,\{d_1,d_2\} \in (1, T],
\\
(\infty, \infty, \infty), & \text{if }  a=0,\, e = 1,\, \{d_1,d_2\} \in (1, T],
\end{cases}\end{align*}}where last branch corresponds to a forced dispatch that occurs when a third truck arrives and the station reaches its capacity.

%{\begin{align}
%\notag &f(s, a, e) =\\
%&\begin{cases}
%(\infty, \infty, \infty), & \text{if } a = 1, \forall e, \\
%(\infty, \infty, \infty), & \text{if } a = 0,\, e = 0,\, d_1 = 1,\, d_2 = \infty, \\
%(\infty, \infty, T), & \text{if } a = 0,\, e = 1,\, d_1 = 1,\, d_2 = \infty, \\
%(\infty, \infty, d_2 - 1), & \text{if } a = 0,\, e = 0,\, d_1 = 1,\, d_2 \leq T, \\
%(\infty, T, d_2 - 1), & \text{if } a = 0,\, e = 1,\, d_1 = 1,\, d_2 \leq T, \\
%(\infty, \infty, d_1 - 1), & \text{if } a = 0,\, e = 0,\, d_1 \in (1,T] ,\, d_2 = \infty, \\
%(\infty, T, d_1 - 1), & \text{if } a = 0,\, e = 1,\,  d_1 \in (1,T],\, d_2 = \infty, \\
%(\infty, d_2 - 1, d_1 - 1), & \text{if } a = 0,\, e = 0,\,  d_2 \leq T,
%\\
%(\infty, \infty, \infty), & \text{if } \forall a,\, e = 1,\, \{d_1,d_2\} \in (1, T].
%\end{cases}\label{eq:transitions}
%\end{align}
%}

%\begin{remark}
%The last branch in Eq.~\eqref{eq:transitions} represents a forced dispatch when a third truck arrives and the station reaches its capacity $L$. This dispatch occurs automatically and does not incur a cost; given $\mathcal{C}_{pt}(L) = 0$. 
%\end{remark}

Since arrivals are stochastic, the transition probabilities of the MDP are given by:
\begin{align*}
\mathbb{P}(s_{n+1} = s' \mid s_n = s,\, a_n = a) =
\begin{cases}
p, & \text{if } s' = f(s, a, 1),\\[2mm]
1-p, & \text{if } s' = f(s, a, 0).
\end{cases}
%\label{eq:tprobabilities}
\end{align*}

\subsubsection{Time Slot Representation} 
The sequence of events within a single time slot is illustrated in Figure~\ref{fig:deadlines3}.\begin{figure}[H]
  \centering
\includegraphics[width=.8\linewidth]{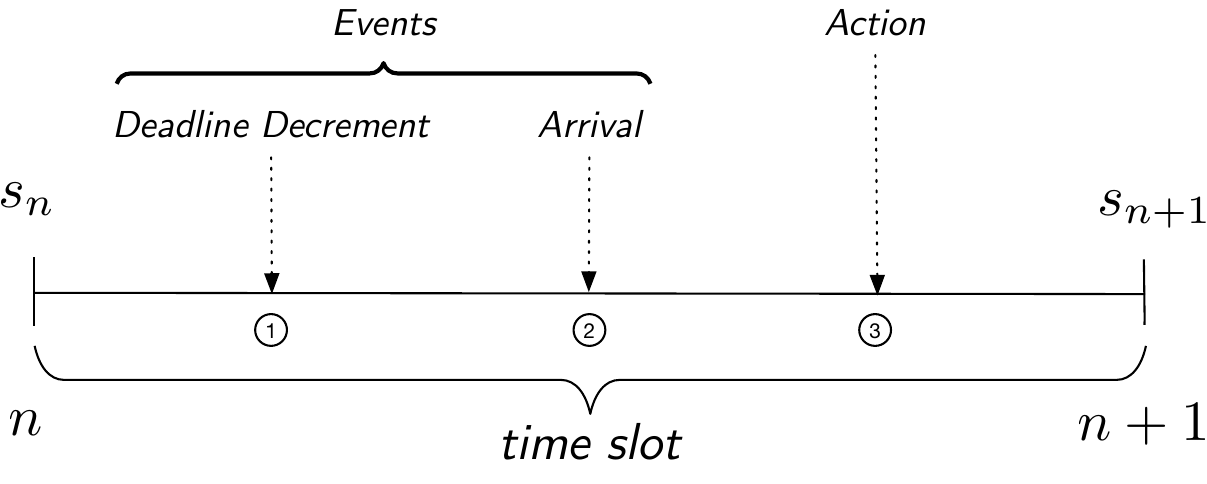}
  \caption{Sequence of events within a time slot.}
  \label{fig:deadlines3}
\end{figure}

\subsubsection{Costs}  Let $\mathcal{C}_{\text{ex}}$ denote the penalty incurred when a truck deadline expires at the station, and let $\omega$ represent the waiting cost per truck per slot. %We denote by $|s_n|$ the number of trucks present in the station at the beginning of slot $n$.

The cost of dispatching a platoon depends on its size $\ell > 0$,  $L$, a scaling factor $\gamma \in (0,1]$, and  $\mathcal{C}_{\text{ex}}$. Formally,\begin{align}
\mathcal{C}_{\text{pt}}(\ell,\, L,\, \gamma,\, \mathcal{C}_{\text{ex}}) = 
\begin{cases}
   0, & \text{if } \ell = L, \\\left(1 - \frac{\ell}{L}\right)  \gamma\, \mathcal{C}_{\text{ex}}, & \text{otherwise},
\end{cases}
\label{eq:cplatoon}
\end{align}

For simplicity, we will write $\mathcal{C}_{\text{pt}}(\ell)$ to refer to the cost in Eq.~\eqref{eq:cplatoon} with fixed $\gamma$, $\mathcal{C}_{\text{ex}}$, and $L$.

 To ensure that releasing full platoons is always preferred, we impose the following cost ordering:
\begin{align}
\mathcal{C}_{\text{pt}}(L) < \ell \cdot \omega < \mathcal{C}_{\text{pt}}(\ell) < \mathcal{C}_{\text{ex}}, \quad 1\le \ell < L. \label{eq:cost_ordering}
\end{align}

\subsubsection{Instantaneous Cost}

The instantaneous cost is computed after the sequence of events is completed and an action is taken. For $s = s_n$, $a = a_n$, and $e = e_n$, this cost is defined as 
\begin{align}
\textsc{ic}(s, a, e) = c_{\text{ex}}(s) + c_{\text{dp}}(s, a, e) + c_{\text{wt}}(s, a, e).
\label{eq:ic}
\end{align}
The three cost components of $\textsc{ic}(s, a, e)$ are defined as follows. Here, $|s|$ denotes the number of finite  deadlines in $s$.
\paragraph{Expiration cost}
\begin{align*}
c_{\text{ex}}(s) = 
\begin{cases}
\mathcal{C}_{\text{ex}}, & \text{if } d_1 = 1, \\
0, & \text{otherwise},
\end{cases}
\end{align*}

\paragraph{Dispatch cost}
\begin{align*}
c_{\text{dp}}(s, a, e) = 
\begin{cases}
\mathcal{C}_{\text{pt}}(|s|), & a = 1,\, e = 0,\, 1 < d_1 \leq T, \\
\mathcal{C}_{\text{pt}}(|s| + 1), & a = 1,\, e = 1,\, 1 < d_1 \leq T, \\
\mathcal{C}_{\text{pt}}(|s| - 1), & a = 1,\, e = 0,\, d_1 = 1, \\
\mathcal{C}_{\text{pt}}(|s|), & a = 1,\, e = 1,\, d_1 = 1, \\
\mathcal{C}_{\text{pt}}(1), & a = 1,\, e = 1,\, d_1 = \infty, \\
0, & \text{otherwise}.
\end{cases}
\end{align*}

\paragraph{Waiting cost}
\begin{align*}
c_{\text{wt}}(s, a, e) = 
\begin{cases}
|s| \, \omega, & a = 0,\, e = 0,\, 1 < d_1 \leq T, \\
(|s| + 1) \, \omega, & a = 0,\, e = 1,\, 1 < d_1 \leq T, \\
(|s| - 1) \, \omega, & a = 0,\, e = 0,\, d_1 = 1, \\
|s| \, \omega, & a = 0,\, e = 1,\, d_1 = 1, \\
\omega, & a = 0,\, e = 1,\, d_1 = \infty, \\
0, & \text{otherwise}.
\end{cases}
\end{align*}
Since $e_n$ is a random variable, $\textsc{ic}(s, a, e)$ is a stochastic quantity. This randomness is explicitly accounted for in Section~\ref{sec:characterization}.

%Our modelling assumptions are formalized .

\subsection{Modelling Assumptions}
\ifbulletlist{\color{blue} 
\begin{enumerate}
    \item Justification of the bang-bang nature of release actions.
    \item Rationale for the cost structure: no cost for platoons of size $L$, penalty for smaller platoons.
    \item Justification of $(\infty, \infty, \infty)$ as the initial state.
\end{enumerate}
}\fi 

The available control actions are \emph{bang-bang}, meaning the controller either releases all trucks or none. This assumption simplifies the analysis by reducing the number of control branches in Eqs.~\eqref{eq:actions}--\eqref{eq:ic}, which is critical given the cardinality of $\mathcal{S}$. Moreover, it ensures that every dispatch maps the system to the empty state $(\infty, \infty, \infty)$, simplifying comparisons by eliminating repeated terms whenever $a = 1$.

This assumption is not overly restrictive: while partial releases could be modelled using a more refined state and cost structure, the current model naturally favours full releases.

To capture the benefits of forming full platoons, we defined the cost $\mathcal{C}_{\text{pt}}(\ell)$ as a decreasing function of  $\ell$, interpreted as a penalty for releasing smaller platoons. This penalty is zero when $L$ trucks are dispatched. It is worth noting that this does not imply that a full platoon incurs no cost, but rather that it carries no additional penalty relative to smaller platoons.

The system is assumed to start empty at slot $n = 0$, i.e., $s_0 = (\infty, \infty, \infty)$. This assumption is without loss of generality, as any initial state with finite deadlines would correspond to trucks that arrived prior to the beginning of the horizon, which are excluded by definition.

\section{Expected Average Cost and Characterization of the Optimal Policy}\label{sec:characterization}
\ifbulletlist{\color{blue} 
\begin{enumerate}
    \item Definition of a average cost of an stationary policy.
   % \item Remarks on the existence of the expected average cost (linear instantaneous cost and finite horizon). Remarks on solving the DP equation for the optimal cost.
    \item Formulation of the expected average instantaneous cost under a stationary policy.
    \item Formal dynamic programming equation.
    \item        Difference between hold and dispatch actions to characterize optimality conditions.
    \item Introduction of two key structural properties: monotonicity and unreachable states.

\end{enumerate}
}\fi

Let $\pi$ be a stationary policy, i.e., a time-independent mapping from states to actions, so that $a_n = \pi(s_n)$. The expected average cost under $\pi$ over a horizon of length $N$ is defined as
\begin{align}
	\dfrac{1}{N}\, \mathbb{E}^\pi \Bigg[ \sum_{n\,=\,0}^{N-1} \textsc{ic}(s_n, a_n, e_n) \Bigg],
	\label{eq:average}
\end{align}
where the expectation is taken with respect to the sample paths induced by~$\pi$.

%starting from state $s_0$ and considering the process $\{e_n\}^{N-1}_{n\,=\,0}$.  

Define $V_{n+1}(s)$ as the optimal expected cost function over the next $n+1$ steps when the system starts in state $s = s_0$. The dynamic programming recursion is given by:\begin{align}
 V_{n+1}(s) &= \min_{a \in \{0,1\}} Q_{n+1}(s,a),
    \label{eq:dp}
\end{align}
where
\begin{align}
	\notag	Q_{n+1}(s,a)  = 
	(1 - p)&\big[ \textsc{ic}(s, a, 0) + V_{n}(f(s, a, 0)) \big] \, + \\
	&\quad p \big[ \textsc{ic}(s, a, 1) + V_{n}(f(s, a, 1)) \big] ,
\end{align}
for all $n \ge 1$, with terminal condition $V_0(s) = 0$ for all $s$.

\begin{figure*}[htpb]
\centering
    \begin{subfigure}{.32\linewidth}
 	\includegraphics[width=\linewidth]{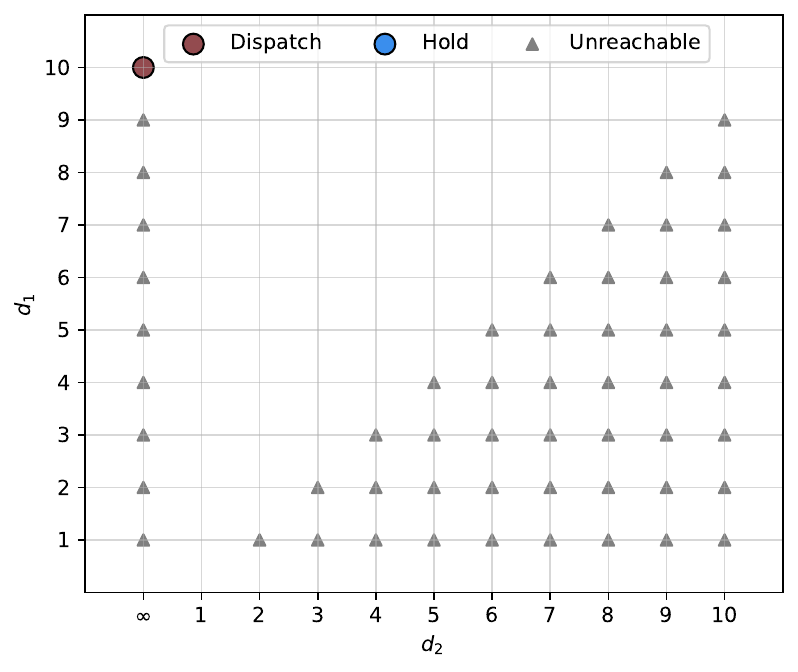}
        \caption{Properties \textbf{c)} and \textbf{d)}.}
        \label{fig:subfig_a}
    \end{subfigure}
    \begin{subfigure}{0.32\linewidth}
        \includegraphics[width=\linewidth]{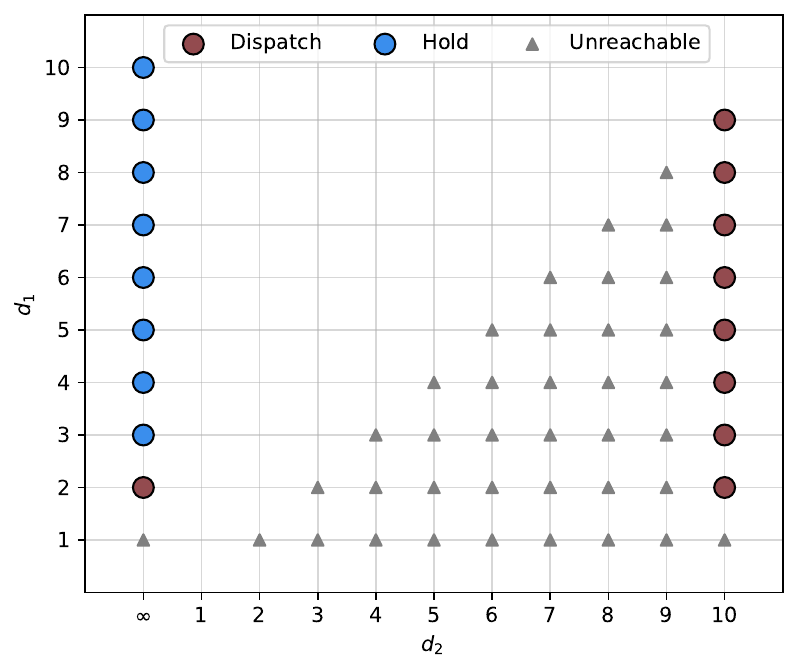}
        \caption{Properties \textbf{a)}, \textbf{c)}, and \textbf{e)}.}
        \label{fig:subfig_b}
    \end{subfigure}
    \begin{subfigure}{0.32\linewidth}
        \includegraphics[width=\linewidth]{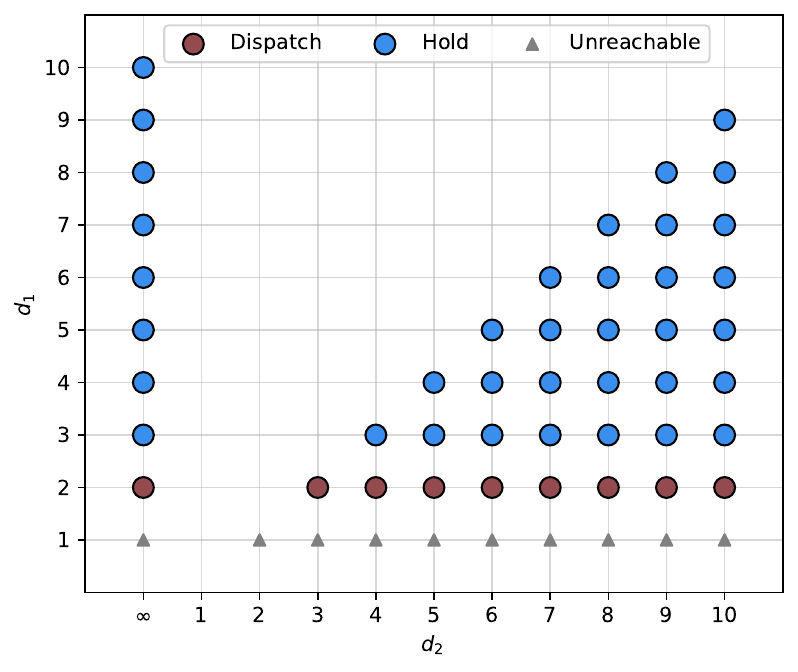}
        \caption{Properties \textbf{a)}, \textbf{b)}, and \textbf{e)}.}
        \label{fig:subfig_c}
    \end{subfigure}
    
\caption{Visualization of monotonicity and unreachability properties of the optimal policy.}    \label{fig:policy_comparison}
\end{figure*}

Now, define the \emph{cost difference} between holding and dispatching at time $n$ in state $s$ as\begin{align}
	\Delta_{n}(s) &= Q_{n}(s, 0) - Q_{n}(s, 1).
\label{eq:delta}
\end{align}

From Eq.~\eqref{eq:delta}, the optimal action $a^\star$ is determined by the sign of $\Delta_{n}(s)$, i.e.,
\begin{align}
a^\star=\pi^\star(s) =
\begin{cases}
0 \text{ (hold)}, & \text{if } \Delta_{n}(s) \le 0,\\
1 \text{ (release)}, & \text{if } \Delta_{n}(s) > 0.
\end{cases}
\end{align}

Next, we show that $\Delta_n(s)$ is \emph{monotone} in $s$. This implies that the optimal action at state $s$ is consistent with the actions at neighbouring states $s' = (d_3', d_2', d_1')$, for $d_i' \in \{d_i, d_i \pm 1\}$.

\subsection{Monotonicity of $\Delta(s)$}
\ifbulletlist
{\color{blue} 
\begin{enumerate}
    \item Main theorem summarizing the three monotonicity properties.
    \item Diagonal and Vertical monotonicity of the hold action.
    \item Monotonicity of the dispatch action.
    \item Example
\end{enumerate}
}\fi

Theorem~\ref{thm:monotonicity} summarizes the key monotonicity properties satisfied by the optimal policy.
\begin{theorem}
\label{thm:monotonicity}
The optimal policy in the platooning model satisfies the following monotonicity properties:
\begin{enumerate}
    \item[\textbf{a)}] \textbf{Tail monotonicity}:  
    If the holding action is optimal at state $s = (\infty,\, \infty,\, d_1)$, then it remains optimal at the shifted-up state $s' = (\infty, \,\infty,\, d_1+1)$.  

    \item[\textbf{b)}] \textbf{Diagonal monotonicity}:  
    If the holding action is optimal at state $s = (\infty,\, d_2,\, d_1)$, then it remains optimal at the diagonally shifted-up state $s' = (\infty,\, d_2+1,\, d_1+1)$.

    \item[\textbf{c)}] \textbf{Dispatch monotonicity}:  
    If the dispatching action is optimal at state $s = (\infty,\, d_2,\, d_1)$, then it remains optimal at all adjacent states with tighter deadlines, namely  $s' \in \{ (\infty,\, d_2,\, d_1 - 1), (\infty,\, d_2 - 1,\, d_1), (\infty,\, d_2 - 1,\, d_1 - 1) \}.$ \end{enumerate}
\end{theorem}
\begin{proof}

\textbf{\textit{a) Tail Monotonicity:}
}
Consider $s = (\infty, \infty, d_1)$ and $s' = (\infty, \infty, d_1 + 1)$. We proceed by induction on n.
\vspace{.2cm}

\noindent\textbf{Base case ($n=1$):}
Since $V_0(s) = 0,$ the value function reduces to the expected instantaneous cost at $n = 1$, i.e.,
\begin{align}
V_1(s) = \min_{a \in \{0,1\}} \big[(1-p),\textsc{ic}(s,a,0) + p,\textsc{ic}(s,a,1) \big].
\end{align}

We claim that action $a = 0$ is optimal for all states $s = (\infty,\, \infty, \, d_1)$ with finite $d_1$ at iteration $n = 1$, i.e.,
\begin{align}
V_1(s) = Q(s,0).
\end{align}

We analyze two cases under the assumption $d_1 \le T$, comparing actions $a = 0$ and $a = 1$:\begin{itemize}
\item \textbf{Case 1: $d_1 = 1$}.  
The earliest truck expires immediately, so both actions incur an expiration penalty. Using Eq.\eqref{eq:ic}, \begin{align}
\notag (1-p)&(\mathcal{C}_{\text{ex}}) + p(\omega + \mathcal{C}_{\text{ex}}) \\ &\le (1-p)(\mathcal{C}_{\text{ex}}) + p\,(\mathcal{C}_{\text{ex}} + \mathcal{C}_{\text{pt}}(1)),
\end{align} 
Cancelling $\mathcal{C}_{\text{ex}}$ from both sides yields
\begin{align}
	p\,\omega \le  p\, \mathcal{C}_{\text{pt}}(1),
\end{align}
which holds under the cost ordering condition $\ell\,\omega < \mathcal{C}_{\text{pt}}(\ell)$ (see Eq.\eqref{eq:cost_ordering}).

\item \textbf{Case 2: $1 < d_1 \le T$}.  
No expiration occurs, and the expected cost inequality becomes
\begin{align}
(1-p)\omega + p\,2\omega \le (1-p)\mathcal{C}_{\text{pt}}(1) + p\,\mathcal{C}_{\text{pt}}(2),\end{align}
which also holds by the ordering in Eq.~\eqref{eq:cost_ordering}.
\end{itemize}

Thus, $a = 0$ is optimal for all $s = (\infty, \infty, d_1)$ at iteration $n = 1$.
For $s'$, the cost is identical when $1 < d_1 \le T$ and smaller when $d_1 = 1$. Hence, $a = 0$ is also optimal at $s'$, i.e.,
\begin{align}
V_1(s') \le V_1(s),
\end{align}
where $V_1(\cdot) = Q_1(\cdot, 0)$.
\vspace{0.2cm}

\noindent\textbf{Inductive step:} 
Assume that $a = 0$ is optimal at state $s$ at iteration $n + 1$, and that
\begin{align}
V_n(s') \le V_n(s).
\end{align}

To prove that $a = 0$ is also optimal for $s'$, we first show that
\begin{align}
Q(s',0) \le Q(s,0), \label{eq:inductiveicv}
\end{align}
for all $s = (\infty,\, \infty,\, d_1)$.

Since $s'$ differs from s only by a one-step increase, \(\textsc{ic}(\cdot,0,e)\) is identical for $s$ and $s'$ when $1 < d_1 \le T$, and smaller for $s'$ when $d_1 = 1$, as state $s$ incurs an expiration cost. Hence,
\begin{align}
\textsc{ic}(s',0,e) \le \textsc{ic}(s,0,e), \quad \forall e \in {0,1}.
\end{align}

Moreover, for each event $e$, the next-state mapping satisfies
\begin{align}
f(s',0,e) = s'^{(e)}, \quad f(s,0,e) = s^{(e)},
\end{align}
where $s'^{(e)}$ corresponds to $s^{(e)}$ with all finite deadlines increased by one.
By the inductive hypothesis $V_n(s') \le V_n(s)$, it follows that
\begin{align}
V_n(f(s',0,e)) \le V_n(f(s,0,e)), \quad \forall e \in {0,1}.
\end{align}
Substituting these relations into yields the desired inequality in Eq.~\eqref{eq:inductiveicv}:
\begin{align}
Q(s',0) \le Q(s,0).
\end{align}

Next, we show that
\begin{align}
Q(s',1) \le Q(s,1), \label{eq:qeq}
\end{align}
for all $s$.

When action $a = 1$ is taken, the truck is dispatched regardless of its remaining deadline. Therefore, both the immediate and future costs are independent of $d_1$, implying  Eq.~\eqref{eq:qeq}.

 Combining Eqs. \eqref{eq:inductiveicv} and \eqref{eq:qeq}, we obtain
\begin{align}
Q(s',0) \le Q(s',1),
\end{align}
which proves that $a = 0$ remains optimal at $s'$ and establishes
\begin{align}
V_{n+1}(s') \le V_{n+1}(s).
\end{align}
 
\noindent \textbf{\textit{b) Diagonal Monotonicity:}
}Consider $s = (\infty, d_2, d_1)$ and $s' = (\infty, d_2 + 1, d_1 + 1)$.
We proceed by induction on n.
\vspace{.2cm}

\noindent\textbf{Base case ($n=1$):}  
We claim that action $a = 0$ is optimal for all states $s = (\infty,\, d_2,\, d_1)$ with $d_1, d_2 \le T$ at iteration $n = 1$, i.e., \begin{align}
V_1(s) = Q(s,0).
\end{align}

We analyze two representative cases, assuming $d_1, d_2 \le T$, and compare actions $a = 0$ and $a = 1$:
\begin{itemize}
\item \textbf{Case 1: $ d_1 = 1 $}.  
\begin{align}
   \notag (1 &- p)(\mathcal{C}_{\text{ex}} + \omega) + p(\mathcal{C}_{\text{ex}} + 2\omega) \leq  \\
    & (1 - p)(\mathcal{C}_{\text{ex}} + \mathcal{C}_{\text{platoon}}(1)) + p(\mathcal{C}_{\text{ex}} + \mathcal{C}_{\text{platoon}}(2)).
\end{align}
   
Cancelling $ \mathcal{C}_{\text{ex}} $ from both sides yields:
    \begin{align}
    (1 - p)\,(\omega - \mathcal{C}_{\text{pt}}(1)) + p\,(2\omega - \mathcal{C}_{\text{pt}}(2)) \leq 0, \label{eq:case1ineq}
    \end{align}
and holds under the cost ordering $\ell\,\omega < \mathcal{C}_{\text{pt}}(\ell)$.

\item \textbf{Case 2: $ 1 < d_1 < d_2 \le T $}.  
No expiration occurs, and the expected cost inequality becomes
\begin{align}
2\omega \le \mathcal{C}_{\text{platoon}}(2),
\end{align}
which again holds by the ordering in Eq.~\eqref{eq:cost_ordering}.
\end{itemize}

Thus, $a = 0$ is optimal for all $s = (\infty, d_2, d_1)$ at iteration $n = 1$. For $s'$, the cost is identical when $1 < d_1, d_2 \le T$, and smaller when $d_1 = 1$, since $s$ then incurs expiration penalties. Hence, $a = 0$ is also optimal at $s'$, i.e.,
\begin{align}
V_1(s') \le V_1(s),
\end{align}
where $V_1(\cdot) = Q_1(\cdot, 0)$.
\vspace{0.2cm}

\noindent\textbf{Inductive step:} 
Assume that $a = 0$ is optimal at state s at iteration $n + 1$, and that
\begin{align}
V_n(s') \le V_n(s).
\end{align}

As before, to prove that $a = 0$ is also optimal for $s'$, we first show that
\begin{align}
Q(s',0) \le Q(s,0), \label{eq:inductive_diag_icv} \end{align}
for all $s = (\infty,\, d_2,\, d_1)$.

Since $s'$ differs from s only by a one-step increase in each finite deadline, \(\textsc{ic}(\cdot,0,e)\) is identical for $s$ and $s'$ when $1 < d_1, d_2 \le T$, and smaller for $s'$ when $d_1 = 1$, as $s$ then incurs expiration costs. Hence,
\begin{align}
\textsc{ic}(s',0,e) \le \textsc{ic}(s,0,e), \quad \forall e \in {0,1}.
\end{align}

Moreover, for each event $e$, the next-state mapping satisfies
\begin{align}
f(s',0,e) = s'^{(e)}, \quad f(s,0,e) = s^{(e)},
\end{align}
where $s'^{(e)}$ corresponds to $s^{(e)}$ with all finite deadlines increased by one.
By the inductive hypothesis $V_n(s') \le V_n(s)$, it follows that
\begin{align}
V_n(f(s',0,e)) \le V_n(f(s,0,e)), \quad \forall e \in {0,1}.
\end{align}
Substituting these relations yields the desired inequality \eqref{eq:inductive_diag_icv}:
\begin{align}
Q(s',0) \le Q(s,0).
\end{align}

Next, we show that
\begin{align}
Q(s',1) = Q(s,1), \label{eq:qeq_diag}
\end{align}
for all s.
When action $a = 1$ is taken, both trucks are dispatched regardless of their remaining deadlines, making both the immediate and future costs independent of $(d_1,d_2)$. Hence, Eq.~\eqref{eq:qeq_diag} holds.

Combining Eqs.~\eqref{eq:inductive_diag_icv} and \eqref{eq:qeq_diag}, we obtain
\begin{align}
Q(s',0) \le Q(s',1),
\end{align}
which proves that $a = 0$ remains optimal at $s'$ and establishes
\begin{align}
V_{n+1}(s') \le V_{n+1}(s).
\end{align}

\noindent\textbf{\textit{c) Dispatch Monotonicity:} } The first iteration at which dispatching ($a = 1$) can become optimal is $n = 2$, since $a = 0$ is optimal for all states in $V_1$, as established by Properties~\textbf{a)} and~\textbf{b)}.

Assume that $a = 1$ is optimal in state $s$ at iteration $n = 2$. Then, by the Bellman update:
\begin{align}
Q_2(s,1) \le Q_2(s,0),
\end{align}
which simplifies to
\begin{align}
\mathcal{C}_{\text{pt}}(2) + V_1(\infty,\infty,\infty)
\le 2\omega + V_1(\infty, T-1, d_1-1). \label{eq:c_base_ineq}
\end{align}

Using Property~\textbf{(b)}, we have $V_1(\infty,\infty,\infty)=p\omega$, hence
\begin{align}
Q_2(s,1) = (1-p)\mathcal{C}_{\mathrm{pt}}(2) + p\omega.
\end{align}

Since dispatching resets the system, for all arrival outcomes $e \in \{0,1\}$:
\begin{align}
\textsc{ic}(s,1,e) \leq \textsc{ic}(s',1,e), \qquad
f(s,1,e) = f(s',1,e).
\end{align}
Thus,
\begin{align}
Q_2(s',1) = Q_2(s,1). \label{eq:Q_equal}
\end{align}

Furthermore, tightening deadlines (moving from $s$ to $s'$) can only increase or preserve the future holding costs:
\begin{align}
V_1(\infty, T-1, d_1-2) \ge V_1(\infty, T-1, d_1-1). \label{eq:c_monotone_future}
\end{align}
Replacing the right-hand side of Eq.~\eqref{eq:c_base_ineq} with the larger value from Eq.~\eqref{eq:c_monotone_future} yields:
\begin{align}
Q_2(s',1) \le Q_2(s',0),
\end{align}
which establishes that $a = 1$ is optimal in $s'$ for $n=2$.

\medskip
\noindent\textbf{Inductive step:}  
Suppose $a = 1$ is optimal in $s$ at iteration $n+1$:
\begin{align}
Q_{n+1}(s,1) \le Q_{n+1}(s,0). \label{eq:c_inductive_assumption}
\end{align}

Dispatching transitions and costs do not depend on deadlines for $d_1 - 1 > 1$, thus:
\begin{align}
Q_{n+1}(s',1) = Q_{n+1}(s,1). \label{eq:c_inductive_equal}
\end{align}

Tightening deadlines can only increase or preserve the holding cost, by monotonicity established in Properties \textbf{a)} and \textbf{b)}, we have:\begin{align}
Q_{n+1}(s',0) \ge Q_{n+1}(s,0). \label{eq:c_inductive_monotone}
\end{align}
Combining Eqs. \eqref{eq:c_inductive_assumption}--\eqref{eq:c_inductive_monotone}, we obtain:
\begin{align}
Q_{n+1}(s',1) \le Q_{n+1}(s',0),
\end{align}
proving that dispatching remains optimal in $s'$ at  $n+1$.\medskip

\noindent\textbf{Conclusion:}
Properties \textbf{a)}, \textbf{b)}, and \textbf{c)} are thus established for all $n \ge 1$, completing the proof.\end{proof} 

In short, Properties \textbf{a)} and \textbf{b)} follow from the fact that relaxing deadlines from $s$ to $s'$ reduces the risk of expirations while increasing the opportunity to form full platoons. More formally, this implies $Q(s, 0) \ge Q(s', 0)$ for all $s$. Therefore, if holding is optimal at a state $s$ with tighter deadlines, it remains optimal at $s'$ when additional slack is available.

Property \textbf{c)} follows from the fact that tightening the deadlines from $s$ to $s'$ increases the urgency of dispatching, since $s'$ is closer to expiration. More formally, this implies that $Q(s, 1) = Q(s', 1)$ for all states $s$ with component  $d_1 \ge 3$.
%Crucially, this does not alter the costs or transitions, i.e., $\textsc{ic}(s, 1, e) = \textsc{ic}(s', 1, e)$ and $f(s, 1, e) = f(s', 1, e)$.

Figure~\ref{fig:policy_comparison} illustrates the above properties across three representative scenarios\footnote{For $L=3$, a two-dimensional representation is possible since $d_3 = \infty$; with finite $d_3$, dispatch is automatic and no decision remains to be taken.}. The figure also highlights the set of \textit{unreachable states}, which are \textbf{valid in principle}---that is, at least one sequence of arrivals and actions could lead to them---but are never visited under $\pi^\star$. These states are discussed next.

%will be discussed next. For simplicity, we set $T=10$, although these properties hold for any $T$.

%In Figure~\ref{fig:policy_comparison}, we present visualizations of optimal policies computed using Eq.~\eqref{eq:dp}, highlighting the monotonicity properties across three scenarios. The figure also illustrates states that are unreachable, which will be discussed next. For simplicity, we assume $T=10$, although these properties hold for any $T$.

%Property \textit{(iii)} characterizes the monotonicity of the dispatching action. Intuitively, if dispatching is optimal at a state with a given deadline, it remains optimal when the deadline tightens. 

%This occurs because the immediate dispatching cost and the value of the resulting next state do not depend on the remaining deadline, while the value of waiting increases as the deadline decreases.

\subsection{Existence of Unreachable States}\ifbulletlist
{\color{blue} 
\begin{enumerate}
\item Definition of unreachable states.
    \item Lemma 1: Unreachable states in the vertical direction.
    \item Lemma 2: Unreachable states in the diagonal direction.
    \item Lemma 3: Unreachable states in both vertical and diagonal directions simultaneously.
\end{enumerate}
}\fi

%In this section, we first formalize the concept of unreachable states and then characterize their properties.

%Let $\mathcal{S}$ denote the set of all valid system states, defined by all feasible combinations of deadlines for waiting trucks; the cardinality $|\mathcal{S}|$ is given in Eq.~\eqref{eq:catalan}. 

%Each state $s \in \mathcal{S}$ is \textbf{valid in principle}, meaning that there exists at least one sequence of arrivals and actions that could lead to it.

%\begin{definition}
A state $s \in \mathcal{S}$ is said \emph{unreachable} if it is never visited along any feasible trajectory 
$\{s_0, s_1, \dots, s_{N-1}\}$ induced by the optimal policy
 $\pi^\star$ and the stochastic arrival process 
$\{e_n\}_{n\,=\,0}^{N-1}$. 

Theorem~\ref{thm:unreachable_states} summarizes the key properties regarding unreachable states satisfied by the optimal policy.

\begin{theorem}
\label{thm:unreachable_states}
The optimal policy in the platooning model satisfies the following properties regarding unreachable states:
\begin{enumerate}
    \item[\textbf{d)}] \textbf{Tail unreachable}:  
        If the dispatching action is optimal at state $s = (\infty,\, \infty,\, d_1)$, then all states of the form $(\infty,\, \infty,\, d_1 - k)$, for $k = 1, \dots, d_1 - 1,$
    are unreachable.
    \item[\textbf{e)}] \textbf{Diagonal unreachable}:  
     If the dispatching action is optimal at state $s = (\infty,\, d_2,\, d_1)$, then all states of the form $(\infty,\, d_2 - k,\, d_1 - k),$ for $k = 1,  \dots, \min(d_1, d_2) - 1$
    are unreachable.
\end{enumerate}
\end{theorem}
\begin{proof}

\textbf{\textit{d) Tail Unreachable:}}
Consider a state $s = (\infty,\, \infty,\, d_1)$ where dispatching is optimal.
By the state dynamics, reaching any state of the form
$(\infty,\, \infty,\, d_1 - k), \quad 1 \le k < d_1$,
requires holding for $k$ consecutive slots so that the last deadline decreases from $d_1$ to $d_1 - k.$

However, under the optimal policy $\pi^\star,$ the system dispatches immediately at $s$, precluding any such waiting.
Therefore, no state $(\infty, \infty, d_1 - k)$ for $k>0$  can be visited, rendering each $(\infty, \infty, d_1 - k)$ unreachable.

\noindent\textbf{\textbf{e) Diagonal Unreachable:}}
Let \(s = (\infty, d_2, d_1)\) with \(d_1 < d_2\) be a dispatching state under \( \pi^\star \). Any state on the diagonal
\[
(\infty, d_2 - k, d_1 - k), 
\quad 1 \le k \le \min(d_1, d_2)-1,
\]
can only be reached from \(s\) by holding for \(k\) slots, causing both deadlines to decrease by \(k\). Since \( \pi^\star \) dispatches immediately at \(s\), such non-dispatching trajectories never occur. Therefore, all these diagonal predecessor states are unreachable under the optimal policy.

\noindent\textbf{Conclusion:} Properties \textbf{d)} and \textbf{e)} are thus established under the optimal policy, completing the proof.
\end{proof}

In short, Properties \textbf{d)} and \textbf{e)} follow from the observation that if reaching  $s'$ requires first passing through a dispatching state $s$, then $s'$ is unreachable. For instance, consider $s = (\infty,\, d_2,\, d_1)$. Due to the deadline decrement mechanism, reaching $s' = (\infty,\, d_2 - 2,\, d_1 - 2)$ requires first visiting the intermediate state $s'' = (\infty,\, d_2 - 1,\, d_1 - 1)$, and ultimately $s$.

More generally, all states along the diagonal of simultaneously decreasing deadlines, i.e., $(\infty,\, d_2 - k,\, d_1 - k$), become unreachable once a dispatch occurs at $s$. In other words, dispatching at $s$ effectively excludes an entire region of the state space from being visited under the optimal policy.

A special case of Property  \textbf{d)} corresponds to the \textit{immediate release} scenario in Figure~\ref{fig:subfig_a}. In this case, it is optimal to dispatch trucks immediately upon arrival, meaning they never wait at the station. Property \textbf{e)} is depicted in Figures~\ref{fig:subfig_b} and~\ref{fig:subfig_c}.

\begin{figure*}[htpb]
\vspace{.2cm}
\centering
    \begin{subfigure}{0.32\linewidth}
        \includegraphics[width=\linewidth]{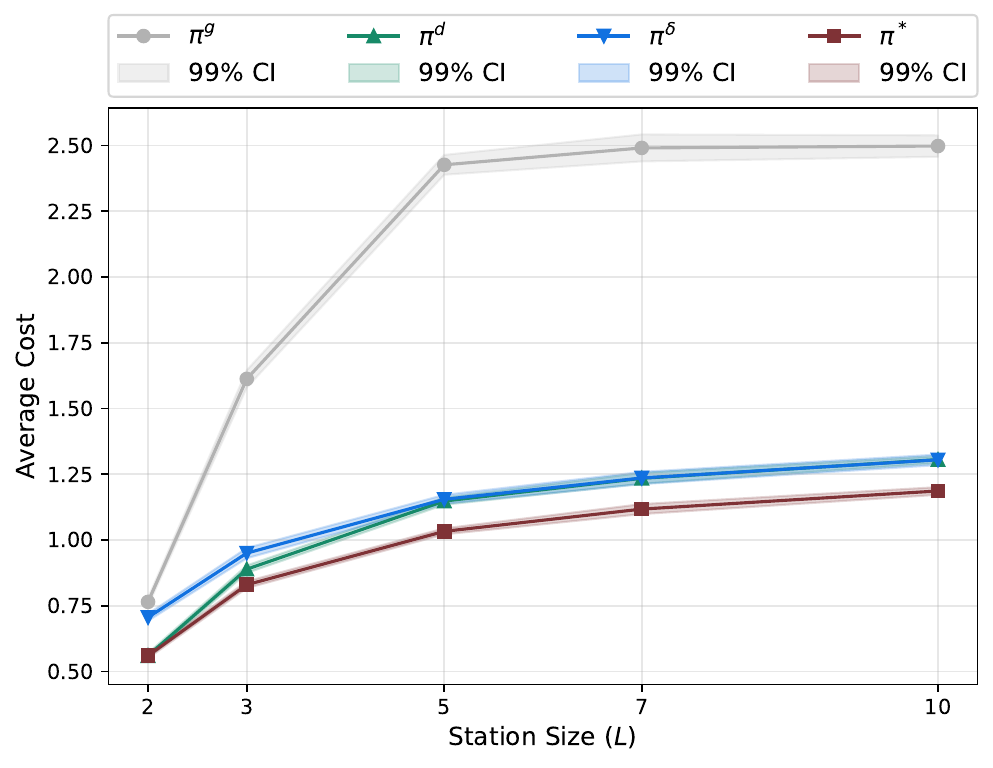}
        \caption{$p=0.1$, $\mathcal{C}_{ex}=15$.}
        \label{res:subfig_a}
    \end{subfigure}
    \begin{subfigure}{0.32\linewidth}
        \includegraphics[width=\linewidth]{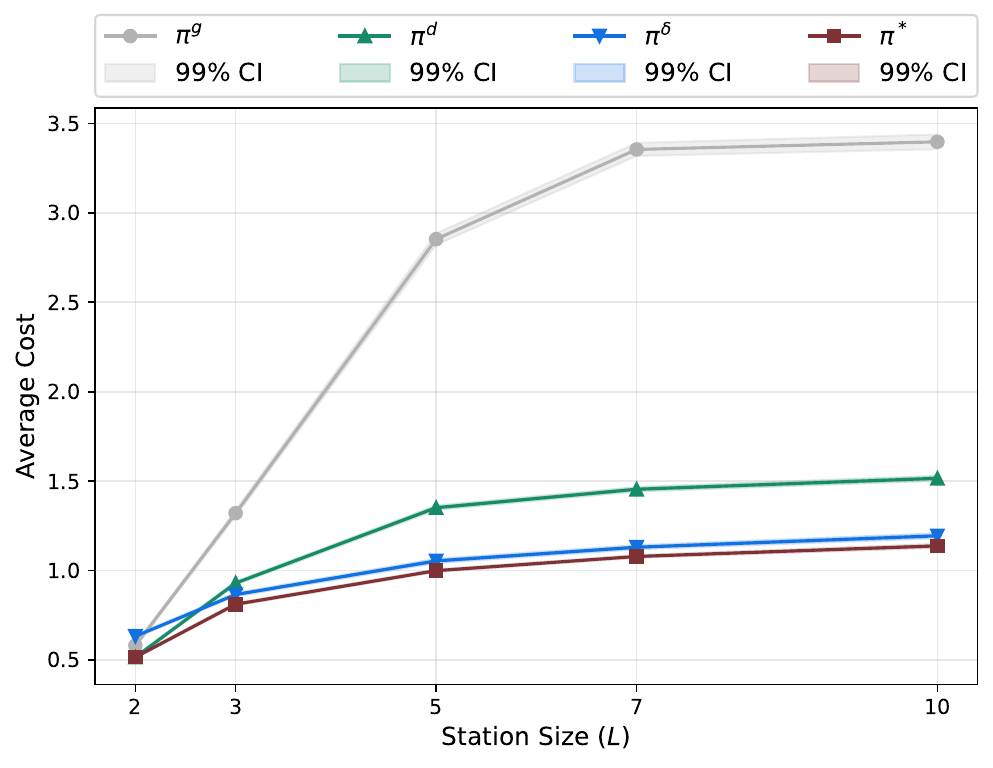}
        \caption{$p=0.2$, $\mathcal{C}_{ex}=7$.}
        \label{res:subfig_b}
    \end{subfigure}
    \begin{subfigure}{0.32\linewidth}
    \includegraphics[width=\linewidth]{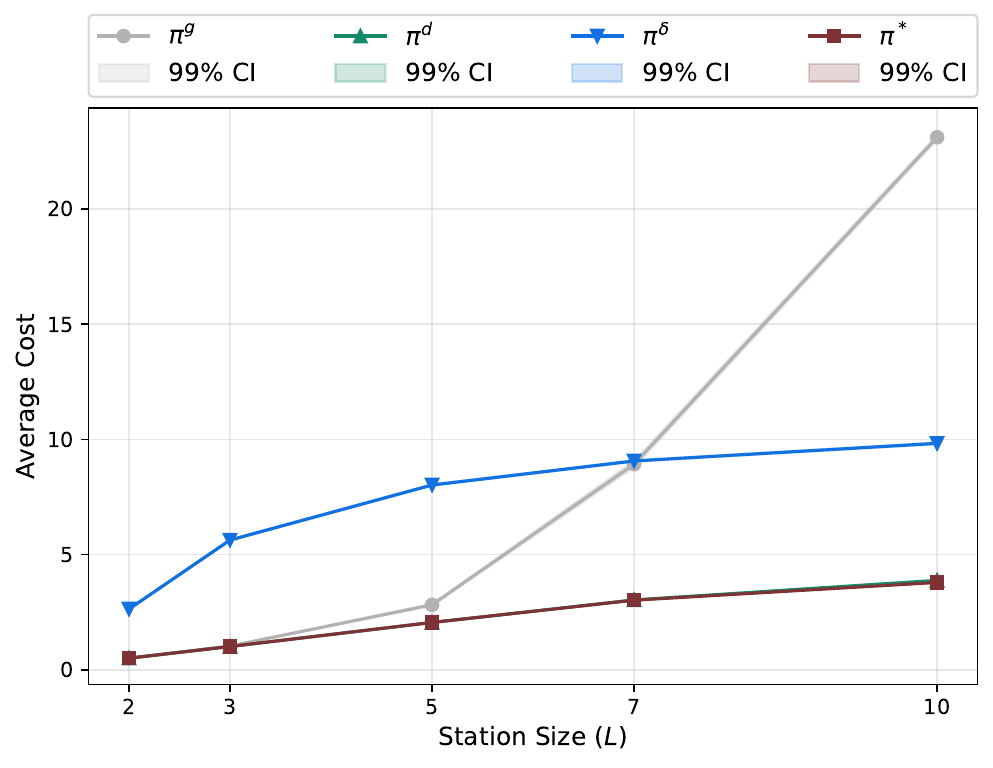}
        \caption{$p=0.6$, $\mathcal{C}_{ex}=30$.}
        \label{res:subfig_c}
    \end{subfigure}
\caption{Discrete event simulations with fixed $T=10$, $\omega=1$, $\gamma=1$.}
\label{fig:des_comparison}\vspace{-.2cm}
\end{figure*}

Properties \textbf{a)} to \textbf{e)} can be shown to extend to arbitrary values of $L$. However, computing $\pi^\star$ (i.e., solving Eq.~\eqref{eq:dp} or its analogous form) quickly becomes computationally intractable---a manifestation of the well-known \textit{curse of dimensionality} in dynamic programming~\cite{bellman1957dynamic}. In the sequel, we introduce three heuristic policies to handle arbitrary station sizes.

%\footnote{For $L>4$,  a low-dimensional representation $\pi^\star$ (as in Figure~\ref{fig:policy_comparison}) is impractical.
 
 %Next, we propose heuristics to address arbitrary values of $L$.
% states that once a dispatch occurs at state $(\infty,\, \infty,\, d_1)$, any subsequent state with tighter deadlines in the tail position becomes unreachable. This occurs because dispatching resets the system to the empty state $(\infty,\,\, \infty,\, \infty)$, preventing the process from naturally evolving toward states with smaller remaining deadlines in that dimension.
%Property \textbf{e)} extends this reasoning to two dimensions. If dispatching is optimal at $(\infty,\, d_2,\, d_1)$, then all states lying along the diagonal line of simultaneously tighter deadlines---that is, $(\infty,, d_2 - k,, d_1 - k)$---become unreachable. In other words, the optimal policy effectively \emph{"cuts off"} an entire region of the state space.

\section{Heuristic Policies}\label{sec:heuristics}
\ifbulletlist
{\color{blue} 
\begin{enumerate}
\item Impact of $L$ and $T$ on state space dimensionality.
    \item Motivation for heuristic policies.
    \item Description of the Greedy policy (release only when station is full).
    \item Description of the Deadline policy (never let trucks expire).
    \item Description of the $\delta$-Deadline policy (configurable threshold parameter $\delta$).

  \end{enumerate}
}\fi

Define three heuristic policies---Greedy, Deadline, and $\delta$-Deep---that leverage the cost ordering in Eq.~\eqref{eq:cost_ordering} and the monotonicity of $\pi^\star$. These policies are designed to achieve near-optimal performance while reducing modelling and computational complexity relative to the DP solution. 

\paragraph{Greedy Policy ($\pi^{\text{g}}$)} This policy dispatches trucks \emph{only} when a full platoon is available. Formally,\begin{align}
\pi^{\text{g}}(s) =
\begin{cases}
1, & \text{if } |s| = L, \\
0, & \text{otherwise}.
\end{cases}
\end{align}
As shown in the next section, $\pi^{\text{g}}$ can incur higher expiration penalties under low arrival rates, since it depends solely on $|s|$ and ignores individual deadlines.

\paragraph{Deadline Policy ($\pi^{\text{d}}$)} This policy leverages Properties \textbf{a)} and \textbf{d)} along with the optimality of full platoons. A dispatch occurs either when $|s| = L$ or when the earliest truck is about to expire (i.e., $d_1 = 2$). Formally,
\begin{align}
\pi^{\text{d}}(s) =
\begin{cases}
1, & \text{if } |s| = L \text{ or } d_1 = 2, \\
0, & \text{otherwise}.
\end{cases}
\end{align}

\paragraph{$\delta$-Deep Policy ($\pi^{\delta}$)} This policy refines $\pi^\text{d}$ by learning an adaptive threshold $\delta$. Specifically, $\pi^\delta$ releases full platoons or when $d_1 = \delta$. Formally, \begin{align}\pi^{\delta}(s) =
\begin{cases}
1, & \text{if } |s| = L \text{ or } d_1 = \delta, \\
0, & \text{otherwise},
\end{cases}
\end{align}
where
$\delta = \mathrm{NeuralNetwork}(L,\, T,\, p,\, \mathcal{C}{\text{ex}},\, \omega,\, \gamma)$
is produced by a \textit{neural network} trained to approximate the optimal policy~$\pi^\star$.
The network is implemented as a feedforward model with six input features corresponding to $(L,\, T,\, p,\, \mathcal{C}_{\text{ex}},\, \omega,\, \gamma)$, followed by two hidden layers with 256 and 512 ReLU-activated neurons, respectively. 

The output layer employs a \textit{softmax} activation that produces a one-hot probability vector over the possible deadline states $d_1 \in {1, \dots, T}$, allowing the network to predict the most likely optimal action associated with each remaining deadline.

Formally, the architecture is defined as
\begin{align*}
	&\text{Input}(6) \rightarrow \text{Dense}(256, \text{ReLU}) \rightarrow \\ &\quad \quad \text{Dense}(512, \text{ReLU}) \rightarrow \text{Dense}(T, \text{Softmax})
\end{align*} 
as illustrated in Figure~\ref{fig:NN}.
\begin{figure}[hbt]
\centering
\includegraphics[width=\linewidth]{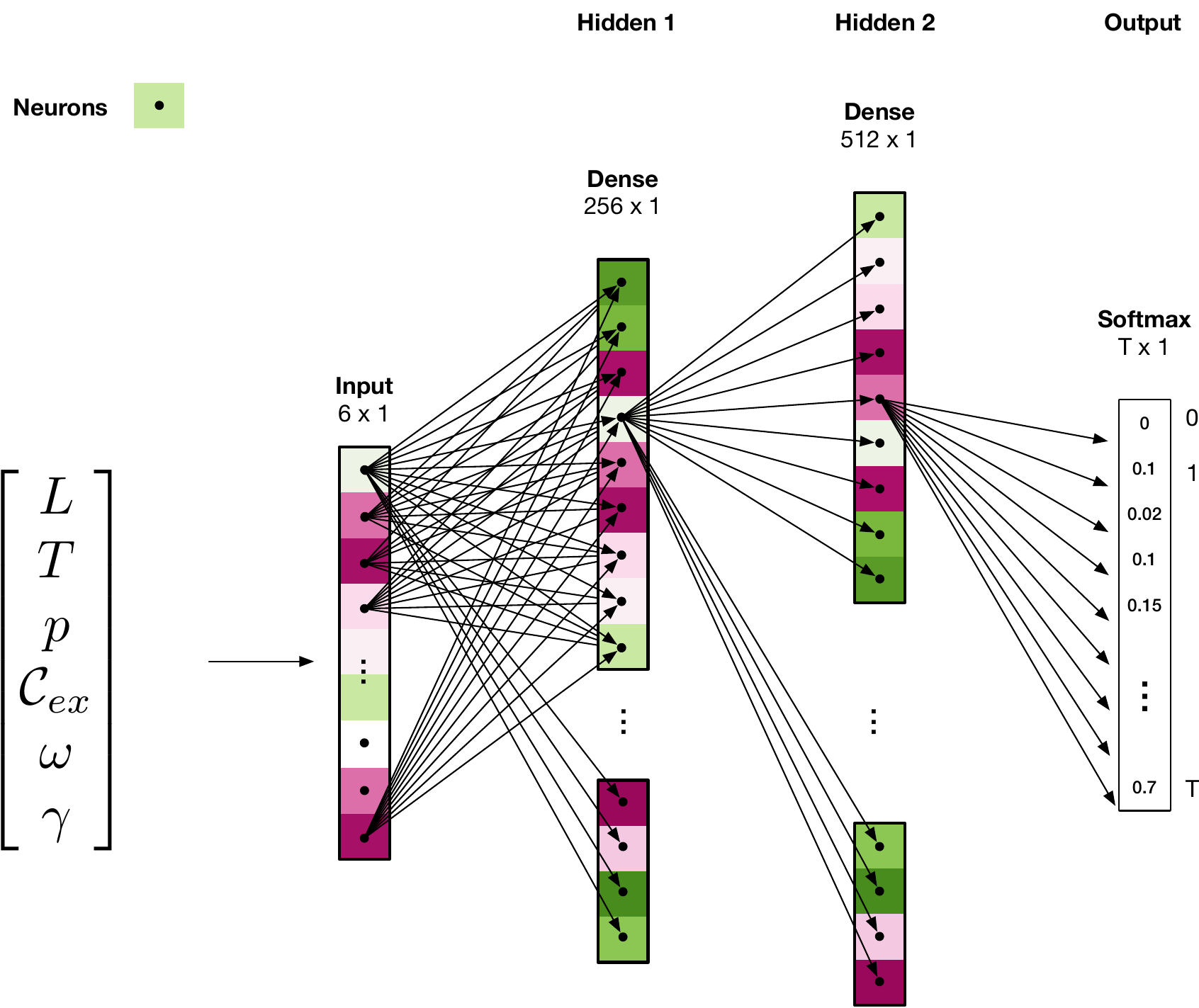}
\caption{Neural network representation of the learned policy.}
\label{fig:NN}
\end{figure}

The output $\delta \in \mathbb{R}^T$ corresponds to a categorical probability distribution over the possible deadline values $d_1 \in {1, \dots, T}$. The most probable threshold, $\arg\max(\delta)$, is selected as the predicted optimal decision for the given system parameters $(L,\, T,\, p,\, \mathcal{C}_{\text{ex}},\, \omega,, \gamma)$.
\begin{remark}
Predicting the optimal $\delta$ alone is insufficient to reconstruct $\pi^\star$, as $d_i$ for $i\in \{2, \dots, T\}$ are disregarded. Rather, $\delta$ identifies the optimal dispatch value of $d_1$, indicating how long a truck should ideally remain at the station.
\end{remark}

%\begin{remark}
%Our heuristics could be enhanced by simultaneously incorporating additional structural properties (e.g., including Property \textbf{b)} in $\pi^\text{d}$). However, our goal is to maintain simple and tractable heuristics, ideally with $\mathcal{O}(1)$ complexity, while achieving comparable performance to $\pi^\star$.
%\end{remark}

Remarkably, both $\pi^\text{g}$ and $\pi^\text{d}$ have constant time complexity, $\mathcal{O}(1)$, as they rely solely on simple conditional checks. In contrast, $\pi^\delta$ is implemented as a lightweight neural network with $|\theta| = 133{,}889$ trainable parameters. Its inference complexity is constant with respect to $L$ and $T$, $\mathcal{O}(1)$, or equivalently linear in the number of parameters, $\mathcal{O}(|\theta|)$. A single forward pass requires roughly $2|\theta| \approx 2.7\times10^5$ floating-point operations, corresponding to only a few  milliseconds on a standard CPU.

Its training complexity, however, scales combinatorially, since both the training and test datasets are generated by repeatedly solving the DP equation---whose complexity itself grows combinatorially---approximately $10^4$ times for instances with $|\mathcal{S}| \le 9 \times 10^4$. After fine-tuning, $\pi^\delta$ achieves over $89\%$ accuracy in predicting $\delta$ for the evaluated scenarios.

\section{Performance Analysis}
\ifbulletlist
{\color{blue} 
\begin{enumerate}
    \item Discrete-event simulation comparing three policy classes.
    \item (If space allows) performance comparison with 4D policies.
 \end{enumerate}
}\fi
\label{sec:performance}

In this section, we assess our heuristics via discrete event simulations, comparing them to the optimal policy obtained through dynamic programming. Six scenarios are considered, each tested with 100 randomly coupled simulation runs, each spanning $10^6$ steps. Coupling ensures that variations in system performance arise solely from differences in the underlying policies~\cite{11161316}. Performance is assessed in terms of average operational cost, along with 99\% confidence intervals.

%In this section, we employ discrete-event simulations to evaluate the performance of the proposed heuristics relative to the optimal policy obtained via dynamic programming. Multiple parameter combinations are examined through 32 randomly coupled\footnote{Coupling ensures that variations in system performance arise solely from differences in the underlying policies~\cite{11161316}.} simulation runs per configuration, each executed over $10^6$ steps. Policies are compared in terms of average operating cost and  corresponding confidence intervals. %, and computational overhead.

Figure~\ref{fig:des_comparison} presents results for $L = 2, 3, 5, 7,$ and $10$, assuming fixed  values for $T=10$, $\omega=1.0$, and $\gamma=1.0$. Across all policies, the confidence intervals are tightly concentrated around the mean, with the largest variation observed for $\pi^\text{g}$.

In Figure~\ref{res:subfig_a}, we illustrate a scenario with infrequent truck arrivals ($p = 0.1$). Notably, $\pi^{g}$ incurs the highest cost due to the low arrival rate, which limits the formation and release of full platoons. As a result, trucks are more likely to expire, triggering penalties and reducing overall performance. For $L = 2$, $\pi^{d}$ performs comparably to $\pi^{\star}$; however, as $L$ increases, its performance approaches that of $\pi^{\delta}$. On average, $\pi^{\delta}$ and $\pi^{d}$ are approximately 8\% less efficient than $\pi^{\star}$, with $\pi^{d}$ offering the advantage of lower computational cost, $\mathcal{O}(1)$.

 %In terms of inference complexity, $\pi^{d}$ and $\pi^{\delta}$ operates in $\mathcal{O}(1)$ time, whereas both $\pi^{\star}$ exhibits combinatorial growth.
\begin{figure*}
\vspace{.2cm}
\begin{subfigure}{0.32\linewidth}
\includegraphics[width=\linewidth]{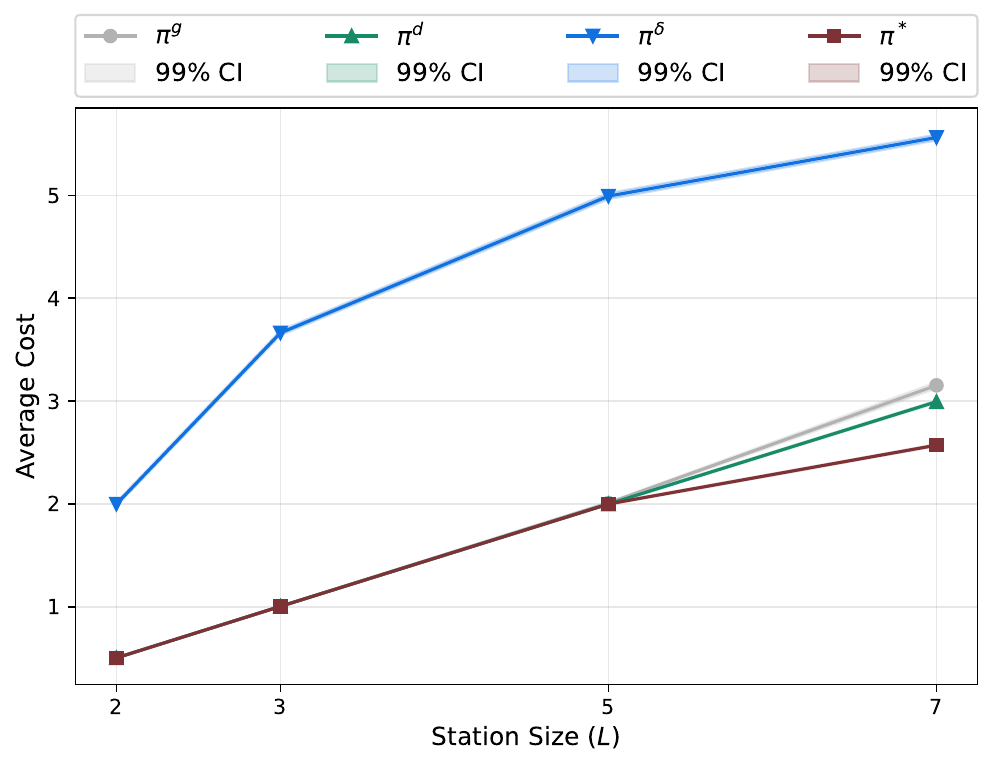}
        \caption{$T=20$, $p=0.5$, $\mathcal{C}_{ex}=25$, $\omega=1$, $\gamma=0.8$, $T=20$.}
        \label{res:2:subfig_a}
    \end{subfigure}
     \begin{subfigure}{0.32\linewidth}
\includegraphics[width=\linewidth]{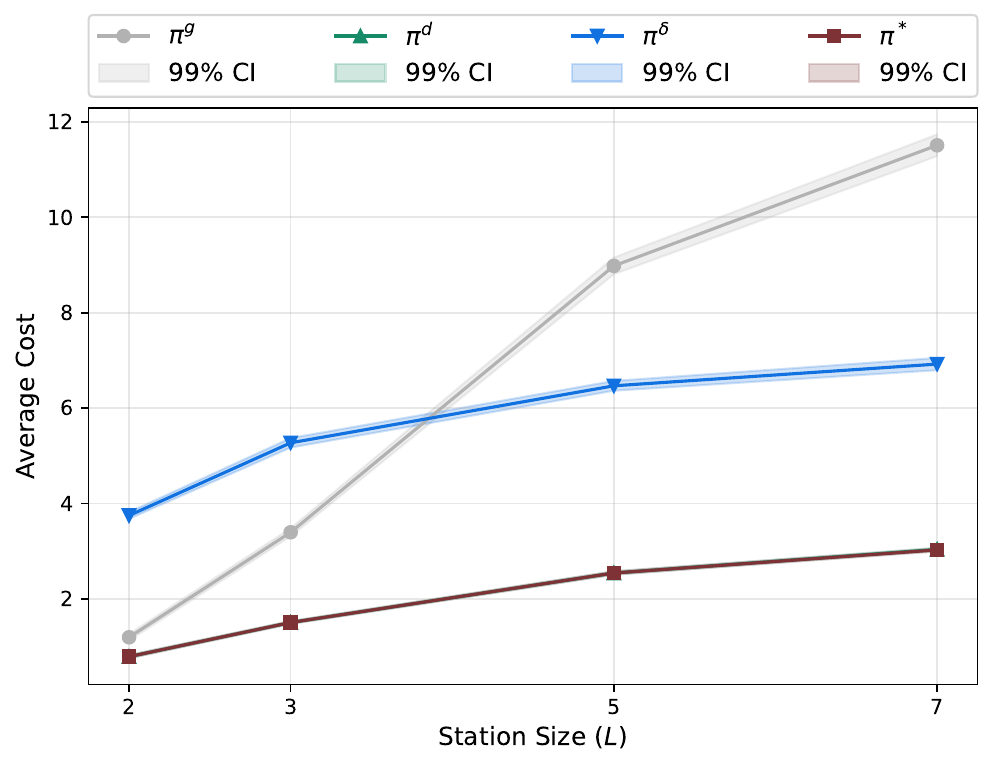}
        \caption{$T=20$, $p=0.1$, $\mathcal{C}_{ex}=100$, $\omega=1$, $\gamma=0.9$.}
        \label{res:2:subfig_b}
    \end{subfigure}
    \begin{subfigure}{0.32\linewidth}
    \includegraphics[width=\linewidth]{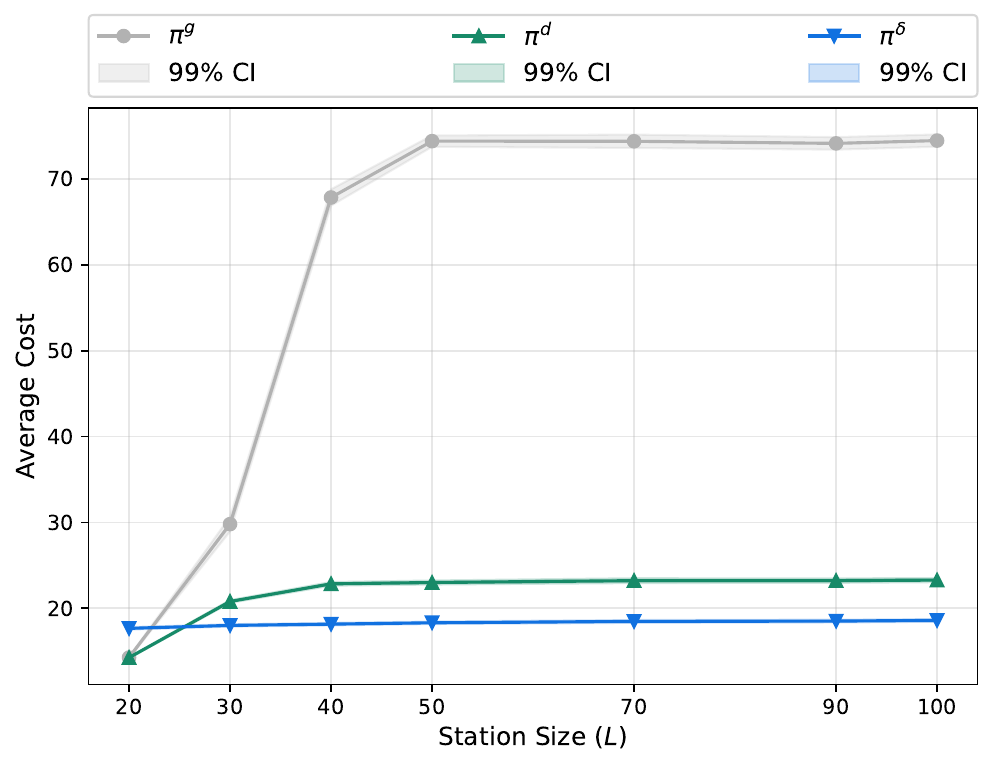}
        \caption{$T=100$, $p=0.3$, $\mathcal{C}_{ex}=100$, $\omega=1.5$, $\gamma=0.8$.}
        \label{res:2:subfig_c}
    \end{subfigure}

\caption{Discrete event simulations for larger $T$.}
\label{fig:des_comparison2}\vspace{-.2cm}
\end{figure*}

In Figure~\ref{res:subfig_b}, with a moderate increase in arrival frequency, $\pi^\delta$ performs comparably to $\pi^\star$, being only about 2\% more costly. In contrast, $\pi^{d}$ is up to 24\% less efficient than $\pi^\delta$. 

In Figure~\ref{res:subfig_c}, however, $\pi^{d}$ performs nearly identically to $\pi^\star$. This behaviour occurs because, at $p = 0.6$, the formation of full platoons becomes highly likely, as three or more arrivals within a window of $T = 10$ slots occur with high probability. Formally, the binomial probability of observing $X \ge 3$ arrivals in 10 slots with $p=0.6$ is $\mathbb{P}(X \ge 3) \approx 0.988$.

Figures~\ref{res:2:subfig_a} and~\ref{res:2:subfig_b} show results for $L = 2, 3, 5,$ and $7$ with $T = 20$, while Figure~\ref{res:2:subfig_c} considers much larger stations ($L = 20$ to $100$) with $T = 100$. The larger $T$ allows trucks to wait longer before expiration, increasing the likelihood of full platoons. Consequently, $\pi^g$ incurs fewer expiration penalties, narrowing its performance gap relative to the other policies.

In Figure~\ref{res:2:subfig_a}, both $\pi^\text{g}$ and $\pi^\text{d}$ achieve optimal performance for $L = 2, 3, 5$, and remain nearly optimal for $L = 7$, with costs approximately $15\%$ higher than $\pi^\star$. Notably, $\pi^\delta$ consistently selects a suboptimal $\delta$ across all values of $L$, likely due to the limited number of training samples for configurations with $L \ge 7$, constrained by $|\mathcal{S}| \le 9 \times 10^4$.

 In Figure~\ref{res:2:subfig_b}, with a high expiration penalty ($\mathcal{C}_{ex} = 100$), $\pi^\text{d}$ achieves optimal performance across all evaluated $L$. In Figure~\ref{res:2:subfig_c}, computing $\pi^\star$ was infeasible due to the size of $|\mathcal{S}|$. Among our heuristics, $\pi^g$ performs significantly worse as $L$ increases. Although additional waiting time allows larger platoons to form, the holding costs may outweigh the benefits of full platoons. $\pi^\delta$ predicts the same threshold $\delta$ across all evaluated station sizes, achieving consistently good performance. $\pi^d$ is roughly 25\% less efficient than $\pi^\delta$.

%\textit{Main Takeaway:} Notably, $\pi^\text{d}$ performs nearly optimally across all evaluated scenarios, which is remarkable given its simple formulation and negligible computational cost. The $\delta$-Deep policy performs well for $T = 10$, but its performance degrades as $T$ increases. We suspect that including examples with larger $|\mathcal{S}|$ in the training and test sets may improve its performance. Finally, $\pi^\text{g}$ achieves near-optimal performance only when $p$ is sufficiently large, allowing trucks to wait long enough for a full platoon to be released.

Overall, the results demonstrate that $\pi^\text{d}$ provides a robust balance between performance and computational efficiency, closely matching $\pi^\star$ across most scenarios, particularly when expiration penalties are high. $\pi^\text{g}$ performs well under frequent arrivals but suffers under low arrival rates, while $\pi^\delta$ adapts to system conditions, achieving near-optimal performance when sufficiently trained, albeit at a higher computational cost.

\section{Conclusion}
\ifbulletlist
{\color{blue} 
\begin{enumerate}
    \item 
 \end{enumerate}
}\fi\label{sec:conclusion}

In this work, we studied the optimal formation of truck platoons at highway stations with capacity $L$ and deadline $T$. For $L = 3$, we proved structural properties of $\pi^\star$, including monotonicity, and identified classes of unreachable states. We then showed how these results extend to larger systems.

Building on the structural insights of $\pi^\star$, we designed near-optimal, low-complexity heuristic policies. Among them, $\pi^d$ consistently balances performance and computational efficiency, $\pi^g$ performs well under frequent arrivals but incurs higher expiration penalties at low arrival rates, and $\pi^\delta$ achieves near-optimal performance when sufficient training data are available, with  a modest increase in modelling complexity.

This work can be extended in several directions. One possibility is to incorporate multiple classes of trucks, each with distinct deadlines, to capture heterogeneous delivery priorities in freight logistics. More sophisticated heuristics could likewise be explored, including those that infer optimal actions for all $d_i$ values rather than only the leading truck. Additionally, the model could be generalized to support partial platoon dispatch, where the controller may release a subset of trucks instead of relying on a strict \textit{all-or-nothing} policy.

\bibliographystyle{IEEEtran}

\bibliography{references} % name of your .bib file

\begin{thebibliography}{10}
\providecommand{\url}[1]{#1}
\csname url@samestyle\endcsname
\providecommand{\newblock}{\relax}
\providecommand{\bibinfo}[2]{#2}
\providecommand{\BIBentrySTDinterwordspacing}{\spaceskip=0pt\relax}
\providecommand{\BIBentryALTinterwordstretchfactor}{4}
\providecommand{\BIBentryALTinterwordspacing}{\spaceskip=\fontdimen2\font plus
\BIBentryALTinterwordstretchfactor\fontdimen3\font minus \fontdimen4\font\relax}
\providecommand{\BIBforeignlanguage}[2]{{%
\expandafter\ifx\csname l@#1\endcsname\relax
\typeout{** WARNING: IEEEtran.bst: No hyphenation pattern has been}%
\typeout{** loaded for the language `#1'. Using the pattern for}%
\typeout{** the default language instead.}%
\else
\language=\csname l@#1\endcsname
\fi
#2}}
\providecommand{\BIBdecl}{\relax}
\BIBdecl

\bibitem{d243b67ca2da4b7e88c224fa5f0ce3af}
R.~Benekohal, Ed., \emph{\BIBforeignlanguage{English (US)}{Traffic Congestion and Traffic Safety in the 21st Century: Challenges, innovations, and opportunities}}.\hskip 1em plus 0.5em minus 0.4em\relax American Society of Civil Engineers, 1997.

\bibitem{Hao04052025}
Y.~Hao, Z.~Chen, J.~Jin, and X.~Sun, ``Joint operation planning of drivers and trucks for semi-autonomous truck platooning,'' \emph{Transportmetrica A: Transport Science}, vol.~21, no.~2, p. 2266041, 2025.

\bibitem{BALADOR2022100460}
A.~Balador, A.~Bazzi, U.~Hernandez-Jayo, I.~{de la Iglesia}, and H.~Ahmadvand, ``A survey on vehicular communication for cooperative truck platooning application,'' \emph{Vehicular Communications}, vol.~35, 2022.

\bibitem{TSUGAWA201341}
S.~Tsugawa, ``An overview on an automated truck platoon within the energy its project,'' \emph{IFAC Proceedings Volumes}, vol.~46, no.~21, pp. 41--46, 2013, 7th IFAC Symposium on Advances in Automotive Control.

\bibitem{ZHANG2024105106}
Y.~Zhang, X.~Chen, J.~Ma, and L.~Yu, ``Environmental impact of autonomous cars considering platooning with buses in urban scenarios,'' \emph{Sustainable Cities and Society}, vol. 101, p. 105106, 2024.

\bibitem{11161316}
T.~S. Gomides, E.~Kranakis, I.~Lambadaris, G.~Shaikhet, and Y.~Viniotis, ``Evaluation of platooning policies using reinforcement learning and correlated arrivals,'' in \emph{ICC 2025 - IEEE International Conference on Communications}, 2025, pp. 5670--5675.

\bibitem{Alvarez01071999}
L.~Alvarez and R.~Horowitz, ``Safe platooning in automated highway systems part i: Safety regions design,'' \emph{Vehicle System Dynamics}, vol.~32, no.~1, pp. 23--55, 1999.

\bibitem{ZHANG2017357}
W.~Zhang, M.~Sundberg, and A.~Karlstrom, ``Platoon coordination with time windows: an operational perspective,'' \emph{Transportation Research Procedia}, vol.~27, pp. 357--364, 2017.

\bibitem{9944383}
W.~Xu, T.~Cui, and M.~Chen, ``Optimizing two-truck platooning with deadlines,'' \emph{IEEE Transactions on Intelligent Transportation Systems}, vol.~24, no.~1, pp. 694--705, 2023.

\bibitem{9102259}
C.~Chen, J.~Jiang, N.~Lv, and S.~Li, ``An intelligent path planning scheme of autonomous vehicles platoon using deep reinforcement learning on network edge,'' \emph{IEEE Access}, vol.~8, pp. 99\,059--99\,069, 2020.

\bibitem{LARSSON2015258}
E.~Larsson, G.~Sennton, and J.~Larson, ``The vehicle platooning problem: Computational complexity and heuristics,'' \emph{Transportation Research Part C: Emerging Technologies}, vol.~60, pp. 258--277, 2015.

\bibitem{Stanley_2015}
R.~P. Stanley, \emph{Catalan Numbers}.\hskip 1em plus 0.5em minus 0.4em\relax Cambridge University Press, 2015.

\bibitem{bellman1957dynamic}
R.~Bellman, \emph{Dynamic Programming}, ser. Rand Corporation research study.\hskip 1em plus 0.5em minus 0.4em\relax Princeton University Press, 1957.

\end{thebibliography}

\end{document}